\documentclass[11pt]{article}

\usepackage{times,amsmath,amsthm,amssymb,graphicx,xspace,epsfig,subfigure,xcolor}\usepackage{multind}
\usepackage{lineno} 
\usepackage{graphicx,multicol}
\usepackage{epic,eepic,epsfig}
\usepackage{pslatex}
\usepackage[vlined]{algorithm2e}
\theoremstyle{plain}
  \newtheorem{theorem}{Theorem}
  \newtheorem{lemma}[theorem]{Lemma}
  \newtheorem{corollary}[theorem]{Corollary}
\newtheorem{proposition}[theorem]{Proposition}
\newtheorem{claim}{Claim}[theorem]

\theoremstyle{definition}
  
  \newtheorem{problem}[theorem]{Problem}
  \newtheorem{conjecture}[theorem]{Conjecture}
 \newtheorem{remark}[theorem]{Remark}

\newcommand{\induce}[2]{\mbox{$ #1 \langle #2 \rangle$}}

\newcommand{\vs}{\vspace{11pt} }

\title{Finding an induced subdivision of a digraph\footnote{This work was done while the first author was on sabbatical at Team Mascotte, INRIA, Sophia Antipolis France whose hospitality is gratefully acknowledged. Financial support from the Danish National Science research council (FNU) (under grant no. 09-066741) is gratefully acknowledged.}}

\author{J\o{}rgen Bang-Jensen\thanks {Department of 
Mathematics and Computer Science,
University of Southern Denmark, Odense DK-5230, Denmark
(email: jbj@imada.sdu.dk).}\and
Fr\'ed\'eric Havet\thanks {Projet Mascotte, I3S (CNRS, UNSA) and INRIA, Sophia Antipolis, France. Partly supported by ANR Blanc AGAPE ANR-09-BLAN-0159.
 (email:Frederic.Havet@sophia.inria.fr).}\and
 Nicolas Trotignon\thanks{ CNRS, LIP -- ENS Lyon, France.  Partially supported by the French
    \emph{Agence Nationale de la Recherche} under reference
    \textsc{anr 10 jcjc 0204 01}. (email: nicolas.trotignon@ens-lyon.fr).   }}

\date{April 17, 2012}

\begin{document}
\maketitle 
\begin{abstract}
We consider the following problem for oriented graphs and digraphs:
  Given an oriented graph (digraph) $G$, does it contain an induced
  subdivision of a prescribed digraph $D$?  The complexity of
  this problem depends on $D$ and on whether $G$ must be an oriented
  graph or is allowed to contain 2-cycles.  We give a number of examples of  polynomial instances as well as several NP-completeness proofs.\\

\end{abstract}

\section{Introduction}
Many interesting classes of graphs are defined by forbidding induced
subgraphs, see~\cite{chudnovsky.seymour:excluding} for a survey.  This
is why the detection of several kinds of induced subgraphs is
interesting, see~\cite{leveque.lmt:detect} where several such problems
are surveyed.  In particular, the problem of deciding whether a graph
$G$ contains, as an induced subgraph, some graph obtained after
possibly subdividing prescribed edges of a prescribed graph $H$ has
been studied.  This problem can be polynomial or NP-complete depending on $H$ and to the set of edges that can be subdivided.  The aim of the
present work is to investigate various similar problems in digraphs,
focusing only on the following problem: given a digraph $H$, is there
a polynomial algorithm to decide whether an input digraph $G$ contains a
subdivision of $H$?

Of course the answer depends heavily on what we mean by ``contain''.
Let us illustrate this by surveying what happens in the realm of
non-oriented graphs.  If the containment relation is the subgraph
containment, then for any fixed $H$, detecting a subdivision of $H$ in
an input graph $G$ can be performed in polynomial time by the Robertson and
Seymour linkage algorithm~\cite{rs:GM13} (for a short explanation of
this see e.g. \cite{bangTCS410}).  But if we want to detect an
\emph{induced} subdivision of $H$ then the answer depends on $H$
(assuming P$\neq$NP).  It is proved in~\cite{leveque.lmt:detect} that
detecting an induced subdivision of $K_5$ is NP-complete, and the
argument can be reproduced for any $H$ whose minimum degree is at
least~4.  Polynomial-time solvable instances trivially exist, such as detecting an
induced subdivision of $H$ when $H$ is a path, or a graph on at most 3
vertices.  But non-trivial polynomial-time solvable instances also exist, such as
detecting an induced subdivision of $K_{2, 3}$ that can be performed
in time $O(n^{11})$ by the Chudnovsky and Seymour's three-in-a-tree
algorithm, see~\cite{chudnovsky.seymour:theta}.  Note that for many
graphs $H$, nothing is known about the complexity of detecting an
induced subdivision of $H$: when $H$ is cubic (in particular when
$H=K_4$) or when $H$ is a disjoint union of 2 triangles, and in many
other cases.

When we move to digraphs, the situation becomes more complicated, even
for the subdigraph containment relation.  All the digraphs we will
consider here are {\it simple}, i.e. they have no loops nor multiple
arcs.  We rely on \cite{livre:digraph} for classical notation and
concepts.  A {\it subdivision of a digraph $D$}, also called a {\it
  $D$-subdivision}, is a digraph obtained from $D$ by replacing each
arc $ab$ of $D$ by a directed $(a,b)$-path.  From the NP-completeness
of the 2-linkage problem, proved by Fortune, Hopcroft and
Wyllie~\cite{fhw:orientedLink}, it is straightforward to construct an
oriented graph $H$ such that deciding whether a given oriented graph
$G$ contains a subdivision of $H$ as a subgraph
is NP-complete.   See
Theorem~\ref{th:strongImNPC}.

Let us now think about the induced subdigraph relation.  An induced
subdigraph of a digraph $G$ which is a subdivision of $D$ is called an
{\it induced subdivision} of $D$.  When $D$ is a digraph, we define:

\vs

\noindent {\sc Problem} $\Pi_D$\\
\underline{Input}: A digraph $G$.\\
\underline{Question}: Does $G$ contain an induced subdivision of $D$?

\vs

In $\Pi_D$, the instance digraph $G$ may have (directed) $2$-cycles, where the
\emph{$2$-cycle} is the digraph $C_2$ on 2 vertices $a, b$ with 2 arcs $ab$
and $ba$.  Because of these 2-cycles, NP-completeness results are
often quite easy to obtain, because no induced directed path can go
through a 2-cycle (which by itself contains a chord).  Hence 2-cycles
are very convenient to force an induced directed path to go through
many places of a large digraph that models an instance of 3-SAT.
This yields NP-completeness results that cover large classes of
detection problems.  See Section~\ref{sec:NPCdigraphs}.  In fact,
it can be easily shown (see
Section~\ref{sec:EasyPoly})  that if $D$ is the disjoint union of {\it spiders} (trees obtained from disjoint
directed paths by identifying one end of each path into a vertex)
and at most one $2$-cycle, then $\Pi_D$ is polynomial-time solvable. 
However, except from those digraphs, we are not aware of any $D$ for which
$\Pi_D$ is polynomial time solvable.
We indeed conjecture that there are none.
As an evidence, we show that if $D$ is an {\it oriented graph}, i.e.
a digraph with no $2$-cycles,  then $\Pi_D$ is NP-complete unless it is
the disjoint union of spiders (see Corollary~\ref{D-in-dig}).

\vs 

It seems that allowing or not allowing 2-cycles is an essential distinction.
Hence we also consider the restricted problem $\Pi'_D$ in which the
input graph $G$ is an oriented graph. 
\vs

\noindent {\sc Problem} $\Pi'_D$\\
\underline{Input}: An oriented graph $G$.\\
\underline{Question}: Does $G$ contain an induced subdivision of $D$?

\vs

Observe that if $\Pi_D$ is
polynomial-time solvable then $\Pi'_D$ is also polynomial-time
solvable.  Conversely, if $\Pi'_D$ is NP-complete then $\Pi_D$ is also
NP-complete.  Hence, NP-completeness results cover less cases for
$\Pi'_D$.  

Similarly to $\Pi_D$, for several $D$'s, $\Pi'_D$ is
solvable by very simple polynomial-time algorithms (See Section~\ref{sec:EasyPoly}).  
However, in this case they are not the only ones.
We could obtain
several digraphs for which $\Pi'_D$ is solvable in polynomial
time with non-trivial algorithms.  

We denote by $TT_3$ the transitive
tournament on 3 vertices $a,b,c$ and arcs $ab,ac,bc$.  In
Subsection~\ref{cherrysec}, we use a variant of Breadth First Search
that computes only induced trees to solve $\Pi'_{TT_3}$ in polynomial time.

We also study oriented paths in Subsection~\ref{subsec:or-path}.
An {\it oriented path} is an orientation of a path.  The {\it length} of an oriented path $P$ is its number of arcs and is denoted $l(P)$.
Its first vertex is called its {\it origin} and its last vertex its {\it terminus}.
The {\it blocks} of an oriented paths are its maximal
directed subpaths. 
 We denote by
$A^-_{k}$ the path on vertices $s_1,s_2,\ldots{},s_k, s_{k+1}$ and arcs $s_2s_1, s_2s_3, s_4s_3,
s_4s_5,\ldots$ and $A^+_{k}$ the path on vertices $s_1,s_2,\ldots{},s_k, s_{k+1}$ and arcs $s_1s_2, s_3s_2, s_3s_4,
s_5s_4,\ldots$. These two paths are the {\it antidirected paths} of length $k-1$.
Observe that $A^-_k$ is the converse of $A^+_k$ (i.e. it is obtained from $A^+_k$ by reversing all the arcs); if $k$ is odd they are isomorphic but the origin and terminus are exchanged.
Clearly,  an oriented path with $k$-blocks can be seen as a subdivision of $A^-_k$ or $A^+_k$.
In particular, paths with one block are the directed paths.
We show that if $P$ is an oriented path with three blocks such that the last one has length one then $\Pi_P$ is polynomial-time solvable.
We also use classical flow
algorithms to prove that $\Pi'_{A^-_4}$ is polynomial-time solvable.

If $D$ is any of the two tournaments on $3$ vertices, namely the
directed $3$-cycle $C_3$ and the transitive tournament $TT_3$, then
$\Pi'_D$ is polynomial time solvable. Hence it is natural to study the
complexity of larger tournaments.  In Section~\ref{sec:NP-tour}, it is
shown that if $D$ is a transitive tournament on more than $3$ vertices
or the strong tournament on $4$ vertices, then $\Pi'_D$ is
NP-complete.

Finally, in Section \ref{sec:remarks}, we point out
several open questions.

\section{Easily polynomial-time solvable problems}\label{sec:EasyPoly}

There are digraphs $D$ for which $\Pi_D$ or $\Pi'_D$ can be easily
proved to be polynomial-time solvable.  For example, it is the case
for the directed $k$-path $P_k$ on $k$ vertices.  Indeed, a
$P_k$-subdivision is a directed path of length at least $k-1$ and an induced
directed path of length at least $k-1$ contains an induced $P_k$. Hence a
digraph has a $P_k$-subdivision if and only if it has $P_k$ as an
induced subdigraph. This can be checked in time $O(n^k)$ by checking
for every set of $k$ vertices whether or not it induces a $P_k$.

A vertex of a digraph is a {\it leaf} if its degree is one, a {\it node} if its out-degree or its in-degree is at least $2$, and a {\it continuity} otherwise, that is if both its out- and in-degree equal~$1$.
A {\it spider} is a tree having at most one node.

\begin{proposition}\label{easy}
If $D$ is the disjoint union of spiders then $\Pi_D$ is polynomial-time solvable. 
\end{proposition}
\begin{proof}
A digraph $G$ contains an induced $D$-subdivision if and only if it contains $D$ as an induced subdigraph.
This can be checked in time $O(n^{|V(D)|})$.
\end{proof}

It is also not difficult to see that $\Pi_{C_2}$ is polynomial-time solvable.

\begin{proposition}\label{2cycle}
$\Pi_{C_2}$ is polynomial-time solvable.
\end{proposition}
\begin{proof}
A subdivision of the directed $2$-cycle is a directed cycle.
In a digraph, a shortest cycle is necessarily induced, hence a digraph has a $C_2$-subdivision if and only if
it is not acyclic.
Since one can check in linear time if a digraph is acyclic or not \cite[Section 2.1]{livre:digraph}, $\Pi_{C_2}$ is polynomial-time solvable.
\end{proof}

Since an oriented graph contains no $2$-cycle, then $\Pi'_{C_2}=\Pi'_{C_3}$.
Similarly to $\Pi_{C_2}$, this problem is polynomial-time solvable.
\begin{proposition}\label{3cycle}
$\Pi'_{C_3}$ is polynomial-time solvable.
\end{proposition}
\begin{proof}An oriented graph contains an induced subdivision of $C_3$ if and only if it is not acyclic.
\end{proof}

Moreover, the following is polynomial-time solvable.

\begin{proposition}\label{easy2}
If $D$ is the disjoint union of spiders and a $C_2$ then $\Pi_D$ is polynomial-time solvable. 
\end{proposition}
\begin{proof}
$D'=D-C_2$ is a collection of spiders. Let $p$ be its order.
For each set $A$ of $p$ vertices, we check if the digraph $G\langle A\rangle$ induced by $A$ is $D'$ and if yes we check if
$G-(A\cup N(A))$ has a directed cycle.
\end{proof}

Similarly,

\begin{proposition}\label{easy3}
If $D$ is the disjoint union of spiders and a $C_3$ then $\Pi'_D$ is polynomial-time solvable. 
\end{proposition}


\section{NP-completeness results for oriented graphs}\label{sec:NPCOriented}

In all proofs below it should be clear that the reductions can be
performed in polynomial time and hence we omit saying this anymore.
Before starting with the NP-completeness proofs, we state a proposition.

\begin{proposition}\label{component}
  Let $D$ be a digraph and $C$ a connected component\footnote{A connected component of a digraph $H$ is a connected component in the underlying undirected graph of $H$.} of $D$.
  If $\Pi_C$ is NP-complete then $\Pi_D$ is NP-complete.
  Similarly, if $\Pi'_C$ is NP-complete then $\Pi'_D$ is NP-complete.
\end{proposition}

\begin{proof}
  Let $D_1, \dots, D_k$ be the components of $D$ and assume that
  $\Pi_{D_1}$ is NP-complete.  To show that $\Pi_D$ is NP-complete, we
  will give a reduction  from $\Pi_{D_1}$ to  $\Pi_{D}$.

  Let $G_1$ be an instance of $\Pi_{D_1}$ and $G$ be the digraph
  obtained from $D$ by replacing $D_1$ by $G_1$.  We claim that $G$
  has an induced $D$-subdivision if and only if $G_1$ has an induced
  $D_1$-subdivision.

  Clearly, if $G_1$ has an induced $D_1$-subdivision $S_1$ then the
  disjoint union of $S_1$ and the $D_i$, $2\leq i\leq k$ is an induced
  $D$-subdivision in $G$.

  Reciprocally, assume that $G$ contains an induced $D$-subdivision
  $S$.  Let $S_i$, $1\leq i\leq k$ be the connected components of $S$
  such that $S_i$ is an induced $D_i$-subdivision.  Set $G_i=D_i$ if
  $i\geq 2$. Then the $G_i$'s are the connected components of $G$.
  Thus $S_1$ is contained in one of the $G_i$'s. If it is $G_1$ then we
  have the result.  Otherwise, it is contained in some other component
  say $G_2=D_2$.  In turn, $S_2$ is contained in some $G_j$. Hence
  $G_j$ contains a $D_1$-subdivision because $S_2$ contains a
  $D_1$-subdivision since $D_2$ contains $S_1$. Thus $G_j$ cannot be
  $G_2$ since $G_2$ already contains $D_1$ and $|S_2|\geq |G_2|$.  If
  $j=1$ then we have the result. If not we may assume that
  $j=3$.  And so on, for every $i\geq 3$, applying the same
  reasoning, we show that one of the following occurs:
  \begin{itemize}
  \item $S_i$ is contained in $G_1$ and thus $G_1$ contains a
    $D_1$-subdivision because $S_i$ did.
  \item $S_i$ is contained in $G_j$ which cannot be any of the $G_i$, $1\leq
    l\leq i$, for cardinality reasons. Hence we may assume that
    $G_j=G_{i+1}$ and that $G_{i+1}$ and hence $S_{i+1}$ contains a $D_1$-subdivision.
  \end{itemize}
  Since the number of components is finite, the process must stop, so $G_1$
  contains an induced $D_1$-subdivision.
\end{proof}

\subsection{Induced $(a,b)$-path in an oriented graph}

Our first result is an easy modification of Bienstock's proof \cite{Bie91} that
finding an induced cycle through two given vertices is NP-complete for
undirected graphs.

\begin{figure}[hbtp]
\begin{center}
\scalebox{0.5}{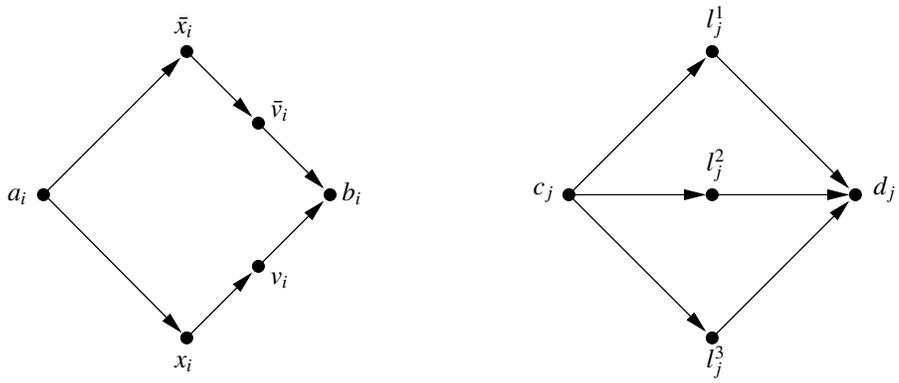}
\end{center}
\caption{The variable gadget $V^1_i$ (left) and  the clause gadget $C^1_j$ (right).}\label{gadgetfig}
\end{figure}

\begin{lemma}\label{inducedpath}
  It is NP-complete to decide whether an oriented graph contains an
  induced $(a,b)$-path for  prescribed vertices $a$ and $b$. 
\end{lemma}
\begin{proof}
  Given an instance $\cal I$ of 3-SAT with variables
  $x_1,x_2,\ldots{},x_n$ and clauses $C_1,C_2,\dots{},C_m$ we first
  create a variable gadget $V^1_i$ for each variable $x_i$,
  $i=1,2,\ldots{},n$ and a clause gadget $C^1_j$ for each clause $C_j$,
  $j=1,2,\ldots{},m$ as shown in Figure \ref{gadgetfig}. Then we form
  the digraph $G_1({\cal I})$ as follows (see Figure
  \ref{inducedpathfig}): Form a chain $U$ of variable gadgets by
  adding the arcs $b_ia_{i+1}$ for $i=1,2,\ldots{},n-1$ and a chain
  $W$ of clause gadgets by adding the arcs $d_jc_{j+1}$,
  $j=1,2,\ldots{},m-1$. Add the arcs $aa_1,b_nc_1,c_mb$ to get a chain
  from $a$ to $b$.  For each clause $C$, we connect the three literal
  vertices of the gadget for $C$ to the variable gadgets for variables
  occuring as literals in $C$ in the way indicated in the figure. To
  be precise, suppose $C_p=(x_i\vee{}\bar{x}_j\vee{}x_k)$, then we add
  the following three 3-cycles $l^1_px_iv_il^1_p$, $l^2_p\bar{x}_j\bar{v}_jl^2_p$
  and $l^3_px_kv_kl^3_p$.  This concludes the construction of $G_1({\cal
    I})$.
    
    \begin{figure}[hbtp]
\begin{center}
\scalebox{0.5}{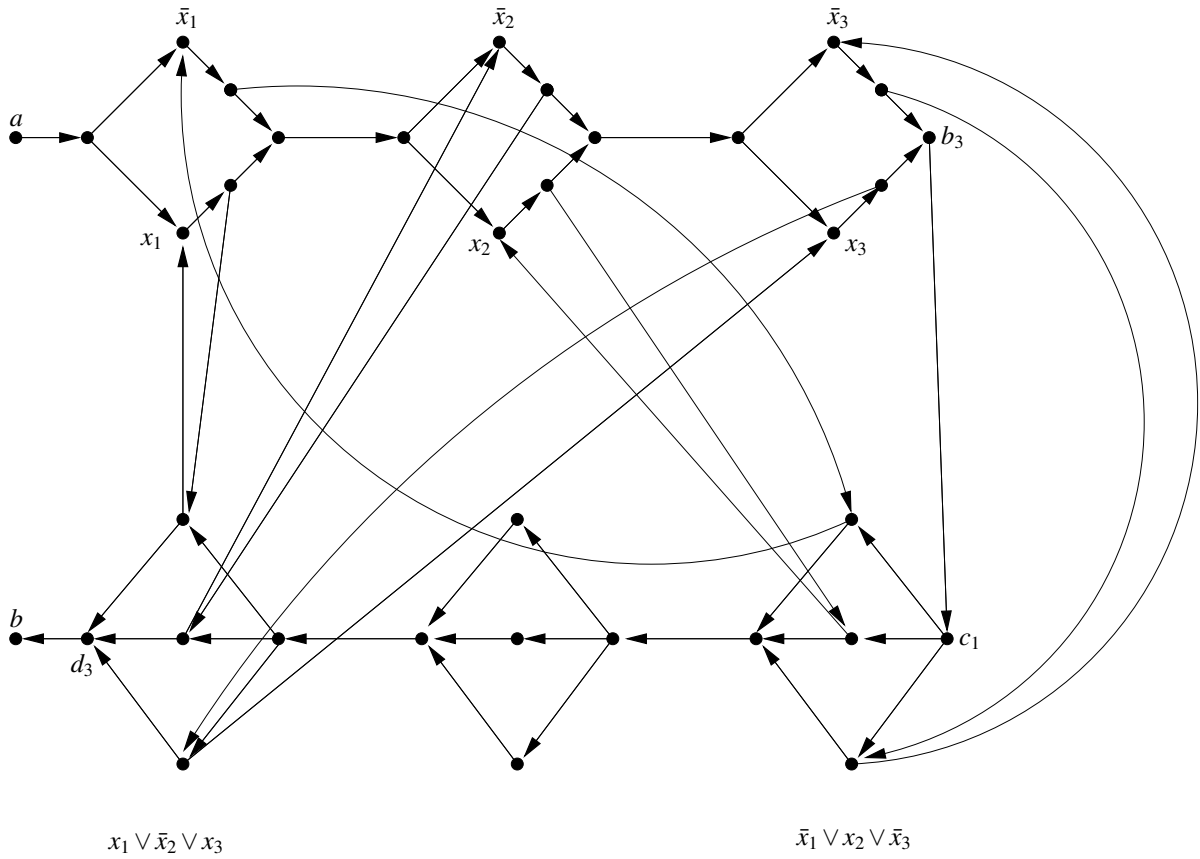}
\end{center}
\caption{The digraph $G_1({\cal I})$ when $\cal I$ has variables $x_1,x_2,x_3$ and three clauses $C_1, C_2, C_3$ where 
$C_1=(\bar{x}_1\vee{}x_2\vee{}\bar{x}_3)$ and $C_3=(x_1\vee{}\bar{x}_2\vee{}x_3)$ (for clarity we do not show the arcs corresponding to $C_2$)}\label{inducedpathfig}
\end{figure}

  We claim that there is an induced directed $(a,b)$-path in $G_1({\cal I})$ if
  and only if ${\cal I}$ is satisfiable. Suppose first that $\cal I$ is
  satisfiable and consider a truth assignment $T$ which satisfies
  $\cal I$.  Now form a directed $(a,b)$-path $P$ by taking the arcs
  $aa_1,c_mb$ and the following subpaths: for each variable $x_i$ take
  the subpath $a_i\bar{x}_i\bar{v}_ib_i$ if $T$ sets $x_i$ true and otherwise
  take the subpath $a_ix_iv_ib_i$. For each clause $C_j$ we fix a
  litteral $l'_j$ which is satisfied by $T$ and take the subpath
  $c_jl'_jd_j$. It is easy to check that $P$ is induced as we navigate
  it to avoid each of the arcs between the variable chain $U$ and the
  clause chain $W$. Suppose now that $Q$ is an induced directed $(a,b)$-path in
  $G_1({\cal I})$. It follows from the construction that $Q$ starts by
  a directed $(a_1,b_n)$-path through all variable gadgets which contains no
  vertices from $W$ and continues with a directed $(c_1,d_m)$-path through all
  clause gadgets which contains no vertices from $U$. This follows
  from the presence of the directed 3-cycles that prevent $Q$ from
  using any of the arcs going from a variable gadget to a clause
  gadget other than the arc $b_nc_1$. Similarly there is no induced directed
  $(c_1,d_m)$-path which contains any vertex from $U$. Now form a
  truth assignment by setting $x_i$ true if and only if $Q$ uses the
  subpath $a_i\bar{x}_i\bar{v}_ib_i$ and false otherwise. Since $Q$ is
  induced, for each clause $C_j$ if $Q$ uses the subpath $c_jl'_jd_j$,
  then we claim that $l'_j$ will be true with the truth assignment just
  described: if $l'_j=x_k$ for some $k$ then since $Q$ is induced the
  presence of the arc $l'_jx_k$ implies that $Q$ uses the path
  $a_k\bar{x}_k\bar{v}_kb_k$ and similarly, if $l'_j=\bar{x}_k$ then $Q$
  uses the path $a_kx_kv_kb_k$ and again $C_j$ is satisfied.
\end{proof}

\subsection{Induced subdivisions of directed cycles}

We first show that for any $k\geq 4$, the problem $\Pi'_{C_k}$ is NP-complete.

\begin{theorem}
  \label{4ormorecycle}
  It is NP-complete to decide whether an oriented graph contains
  an induced subdivision of a fixed directed cycle of length at least 4.
\end{theorem}

\begin{proof}
  Given an instance $\cal I$ of 3-SAT with variables
  $x_1,x_2,\ldots{},x_n$ and clauses $C_1,C_2,\dots{},C_m$ we form the
  digraph $G^*_1({\cal I})$ from $G_1({\cal I})$ which we defined above by
  adding the arc $ba$.  
  
  Let $C$ be an  induced cycle of $G^*_1({\cal I})$.
Since the variable chain $U$ and the clause chain $W$ are both acyclic, $C$ must contain an arc
with tail $l$ in $W$ and head $y$ in $U$. If $ly \neq ba$, then there exists $i$ such that $y\in \{x_i, \bar{x}_i\}$ and so $C=lx_iv_il$ or   $C=l\bar{x_i}\bar{v}_il$ by construction of $G^*_1({\cal I})$. Hence every induced directed cycle of length at least 4 contains the arc $ba$. Thus $G^*_1({\cal
    I})$ has an induced cycle of length at least 4 if and only if $G_1({\cal
    I})$
  has an induced directed $(a,b)$-path. As shown in the proof of Lemma
  \ref{inducedpath} this is if and only if $\cal I$ is satisfiable.
\end{proof}

\begin{theorem}\label{cycle-d2}
  Let $D$ be an oriented graph containing an induced directed cycle of
  length at least 4 with a vertex of degree\footnote{The degree of a vertex $v$ in a digraph is the number of arcs with one end in $v$, that is, the sum of the in- and out-degree of $v$.}~2.  It is NP-complete to
  decide whether a given oriented graph contains an induced
  subdivision of $D$.
\end{theorem}

\begin{proof}
  Let $D$ be given and let $\cal I$ be an arbitrary instance of
  3-SAT. Fix an induced directed cycle $C$ of length at least 4 in $D$
  and fix an arc $uv$ on $C$ such that $u$ is of degree~2. Let
  $G'_1({\cal I})$ be the oriented graph that we obtain by replacing the
  arc $uv$ by a copy of $G_1({\cal I)}$ and the arcs $ua,bv$. We claim
  that $G'_1({\cal I})$ contains an induced subdivision of $D$ if and
  only if $\cal I$ is satisfiable (which is if and only if $G_1({\cal
    I})$ contains an induced directed $(a,b)$-path). 
    
    Clearly, if $G_1({\cal I})$
  has an induced directed $(a,b)$-path, then we may use the concatention of this path with $ua$ and $bv$ instead of the
  deleted arc $uv$ to obtain an induced $D$-subdivision in $G'_1({\cal I})$
  (the only subdivided arc will be $uv$). 
 
  Conversely, suppose that $G'_1({\cal I})$ contains an induced subdivision $D'$ of $D$. Clearly
  $D'$ has at least as many vertices as $D$ and thus must contain at least one vertex $z$ of $V(G_1({\cal I}))$. 
  Since $u$ is of degree 2, the digraph  $D\setminus uv$ has fewer induced directed cycles of length at least 4
  than $D$. (Note that the fact that $u$ is of degree 2 is important:
  if $u$ has degree more than 2, deleting $uv$ could create new
  induced directed cycles. )
  Thus $z$ must be on a cycle of length at least
  4 in $D'$.   But this and the fact that $G_1({\cal I})$
  has no induced directed cycle of length at least 4 implies that $G'_1({\cal I})$
  contains an induced directed $(a,b)$-path (which passes through $z$).
\end{proof}

We move now to the detection of induced subdivisions of digraphs $H$ when $H$ is  the disjoint union of one or more  directed cycles, all of length 3.  If there is just one cycle in $H$, the problem is
polynomial-time solvable by Proposition~\ref{3cycle}.  But from two on, it becomes NP-complete. We need results on
the following problem.  \vs

\noindent {\sc Problem} DIDPP\\
\underline{Input}: An acyclic digraph $G$ and two vertex pairs
$(s_1,t_1), (s_2,t_2)$.  Moreover, there is no directed path from
$\{s_2, t_2\}$  to $\{s_1, t_1\}$.\\
\underline{Question}: Does $G$ have two paths $P_1$, $P_2$
such that $P_i$ is a directed $(s_i,t_i)$-path, $i=1, 2$,  and $G\langle V(P_1)\cup 
 V(P_2)\rangle$ is the disjoint union of  $P_1$ and $P_2$?

\vs

Problem $k$-DIDPP was shown to be NP-complete by
Kobayashi~\cite{Kob09} using a  proof similar to Bienstock's proof 
in \cite{Bie91}.

\begin{theorem}\label{disjoint-cycles}
  Let $D$ be the disjoint union of two directed cycles with no arcs
  between them.  Then $\Pi'_D$ is NP-complete.
\end{theorem}

\begin{proof}
  Let $G$ be an instance of DIDPP and $H$ the oriented graph
  obtained from it by adding new vertices $u_1, u_2$ and the arcs
  $t_1u_1$, $u_1s_1$, $t_2u_2$ and $u_2s_2$.  Since $G$ was acyclic it
  is not difficult to see that $H$ is a yes-instance of $\Pi'_D$ if
  and only if $G$ is a yes-instance of DIDPP.
\end{proof}

\section{NP-completeness results for digraphs}\label{sec:NPCdigraphs}

\begin{figure}[hbtp]
\begin{center}
\scalebox{0.5}{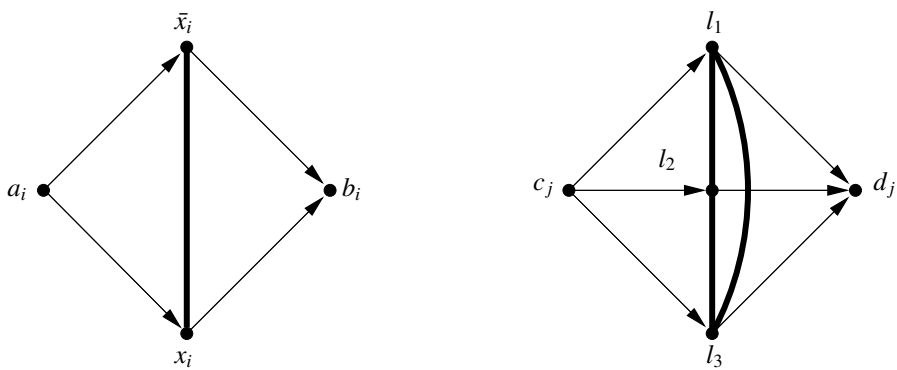}
\end{center}
\caption{The variable gadget $V^2_i$ (left)  and the clause gadget $C^2_j$ (right).
Unoriented bold edges represent $2$-cycles.
}\label{gadget-G3fig}
\end{figure}

\begin{theorem}\label{cycles-dig}
Let $k\geq 3$ be an integer.
Then $\Pi_{C_k}$ is NP-complete.
\end{theorem}
\begin{proof}
Reduction from $3$-SAT.
Let ${\cal I}$ be an instance of $3$-SAT with variables
  $x_1,x_2,\ldots{},x_n$ and clauses $C_1,C_2,\dots{},C_m$. We first
  create a variable gadget $V^2_i$ for each variable $x_i$,
  $i=1,2,\ldots{},n$ and a clause gadget $C^2_j$ for each clause $C_j$,
  $j=1,2,\ldots{},m$ as shown in Figure \ref{gadget-G3fig}. Then we form
  the digraph $G_2({\cal I})$ as follows (see Figure
  \ref{G3fig}): Form a chain $U$ of variable gadgets by
  adding the arcs $b_ia_{i+1}$ for $i=1,2,\ldots{},n-1$ and a chain
  $W$ of clause gadgets by adding the arcs $d_jc_{j+1}$,
  $j=1,2,\ldots{},m-1$. Add the arcs $aa_1,b_nc_1,c_mb$ to get a chain
  from $a$ to $b$.  For each clause $C$, we connect the three literal
  vertices of the gadget for $C$ to the variable gadgets for variables
  occuring as literals in $C$ in the following way. Suppose $C_p=(x_i\vee{}\bar{x}_j\vee{}x_k)$, then we add
  the following three 2-cycles $l^1_px_il^1_p$, $l^2_p\bar{x}_jl^2_p$
  and $l^3_px_kl^3_p$.  This concludes the construction of $G_2({\cal
    I})$. See Figure~\ref{G3fig}.
    
    \begin{figure}[hbtp]
\begin{center}
\scalebox{0.5}{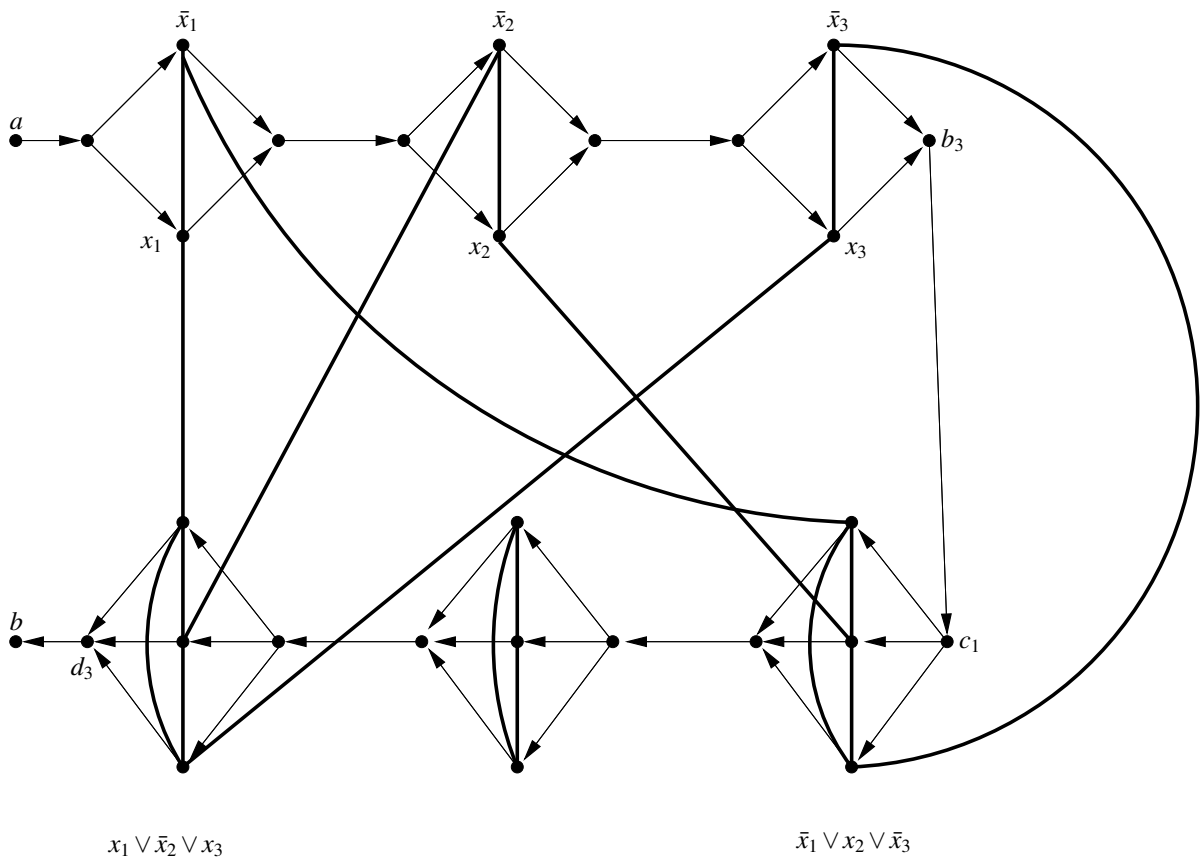}
\end{center}
\caption{The digraph $G_2({\cal I})$ when $\cal I$ has variables $x_1,x_2,x_3$ and three clauses $C_1, C_2, C_3$ where 
$C_1=(\bar{x}_1\vee{}x_2\vee{}\bar{x}_3)$ and $C_3=(x_1\vee{}\bar{x}_2\vee{}x_3)$ (for clarity we do not show the arcs corresponding to $C_2$)}\label{G3fig}
\end{figure}

Similarly to the proof of Lemma~\ref{inducedpath}, one can show that there is an induced directed $(a,b)$-path in $G_2({\cal I})$ if
  and only if ${\cal I}$ is satisfiable.

Let $G^k_2({\cal I})$ be the digraph obtained from $C_k$ by replacing one arc $ab$ by $G_2({\cal I})$. It is easy to check that $G_2({\cal I})$ has no induced cycle of length at least $3$.
Hence $G^k_2({\cal I})$ has an induced directed cycle of length $k$ if and only if $G_2({\cal I})$ has an induced directed $(a,b)$-path.
Hence by Lemma~\ref{inducedpath}, $G^k_2({\cal I})$ has an induced $D$-subdivision if and only if ${\cal I}$ is satisfiable.
\end{proof}

A {\it branch} is a directed walk such that all the vertices are distinct except possibly its ends,
its ends are nodes or leaves and all its internal vertices are continuities.
A branch is {\it central} if its two ends are nodes.

The {\it skeleton} of a multidigraph $D$ is the digraph whose vertices are the nodes and leaves in $D$ and in which  $ab$ is an arc if and only if there is a directed $(a,b)$-branch in $D$.
Observe that a skeleton may have loops and multiple arcs.
Clearly, any subdivision of $D$ has the same skeleton as $D$.

\begin{theorem}\label{central}
Let $D$ be an oriented graph. If $D$ contains a central branch, then $\Pi_D$ is NP-complete.
\end{theorem}
\begin{proof}
Reduction from $3$-SAT.
Let ${\cal I}$ be an instance of $3$-SAT. 
Let $B$ be a central branch with origin $a$ and terminus $c$.
Let $G^D_2({\cal I})$ be the digraph obtained from $D$ by replacing the first arc $ab$ of $B$ by 
$G_2({\cal I})$.

Clearly if $G_2({\cal I})$ has an induced directed  $(a,b)$-path $P$, then the union of $P$ and $D\setminus ab$
is a $D$-subdivision (in which only $ab$ is subdivided) in  $G^D_2({\cal I})$.

Conversely, assume that $G^D_2({\cal I})$ contains an induced
$D$-subdivision $S$.  It is easy to check that no vertex in
$V(G_2({\cal I}))\setminus \{a,b\}$ can be a node of $S$ (the 2-cycles prevent this).  Then since
$S$ has the same skeleton as $D$, $a$ and $b$ are nodes of $S$.  In
addition, since the number of central branches in $D\setminus ab$ is
one less than the number of central branches in $D$, one central
branch of $D$ must use vertices of $G_2({\cal I})$. Thus, there is an
induced directed $(a,b)$-path in $G_2({\cal I})$.

Hence $G^D_2({\cal I})$ has an induced $D$-subdivision if and only if $G_2({\cal I})$ has an induced directed  $(a,b)$-path and thus if and only if ${\cal I}$ is satisfiable.
\end{proof}

\begin{corollary}\label{D-in-dig}
Let $D$ be an oriented graph. Then  $\Pi_D$ is NP-complete unless $D$ is  the disjoint union of spiders.
\end{corollary}
\begin{proof}
Let $D$ be an oriented graph.
If one of its connected components is neither a directed cycle nor a spider, then
it contains at least one  central branch. So $\Pi_D$ is NP-complete by Theorem~\ref{central}.

If one of the components is directed cycle of length at least $3$, then $\Pi_D$ is NP-complete by Theorem~\ref{cycles-dig} and Proposition~\ref{component}. 

Finally, if all its connected components are spiders  then 
$\Pi_D$ is polynomial-time solvable according to Theorem~\ref{easy3}.
\end{proof}

\vspace{12pt}

We believe that Corollary~\ref{D-in-dig} can be generalized to digraphs.

\begin{conjecture}\label{D-in-dig-conj}
Let $D$ be a digraph. Then  $\Pi_D$ is NP-complete unless $D$ is  the disjoint union of spiders and at most one $2$-cycle.
\end{conjecture}

As support for this conjecture, we give some other
digraphs $D$ (which are not oriented graphs), for which $\Pi_D$ is
NP-complete. In particular, when $D$ is the {\it lollipop}, that is
the digraph $L$ with vertex set $\{x,y,z\}$ and arc set $\{xy, yz,
zy\}$.  Note that the lollipop seems to be the simplest digraph that is not an oriented graph nor a $C_2$.
So it should be an obvious candidate for a further polynomial
case if one existed.

\begin{theorem}\label{lollipop}
Deciding if a digraph contains an induced subdivision of the lollipop is NP-complete.
\end{theorem}
\begin{proof}
Reduction from $3$-SAT.
Let ${\cal I}$ be an instance of $3$-SAT with variables
  $x_1,x_2,\ldots{},x_n$ and clauses $C_1,C_2,\dots{},C_m$. We first
  create a variable gadget $V^3_i$ for each variable $x_i$,
  $i=1,2,\ldots{},n$ and a clause gadget $C^3_j$ for each clause $C_j$,
  $j=1,2,\ldots{},m$ as shown in Figure \ref{lolligadget}. Then we form
  the digraph $G_3({\cal I})$ as follows: Form a chain $U$ of variable gadgets by
  adding the arcs $b_ia_{i+1}$ for $i=1,2,\ldots{},n-1$ and a chain
  $W$ of clause gadgets by adding the arcs $d_jc_{j+1}$,
  $j=1,2,\ldots{},m-1$. Add the arcs $aa_1,b_nc_1,c_mb$ to get a chain
  from $a$ to $b$.  For each clause $C$, we connect the three literal
  vertices of the gadget for $C$ to the variable gadgets for variables
  occuring as literals in $C$ in the way indicated in the figure.

\begin{figure}[hbtp]
\begin{center}
\scalebox{0.6}{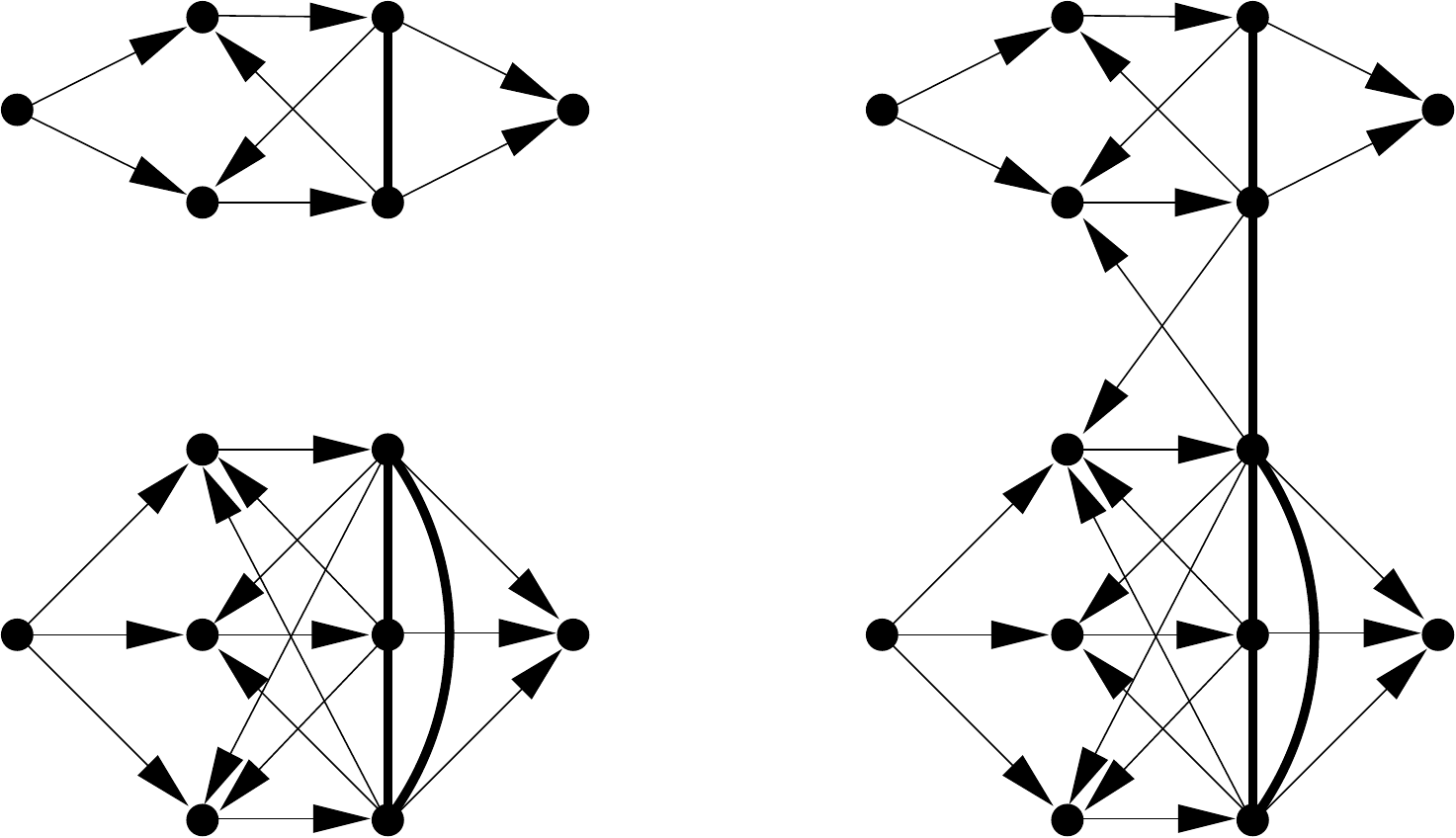}
\caption{The variable gadget $V^3_i$, (top left), the clause gadget $C^3_i$ (bottom left) and the way to connect them in  $G_3({\cal I})$ (right).  Bold unoriented edges represent $2$-cycles. Only the connection for one variable gadget and one clause gadget is shown and the general strategy for connecting variable and clause gadgets is the same as in $G_1(I)$ (Figure \ref{inducedpathfig}).}
\label{lolligadget}
\end{center}
\end{figure}

Similarly to the proof of Lemma~\ref{inducedpath}, one can check that there is an induced directed $(a,b)$-path in $G_3({\cal I})$ if and only if ${\cal I}$ is satisfiable.

The digraph $G_3^L({\cal I})$ is obtained from $L$ and $G_3({\cal I})$ by deleting the arc $yz$ and adding the arcs $ya$ and $bz$.

It is easy to see that $G_3({\cal I})$ has no induced directed cycle of length $3$ and that no $2$-cycle is contained in an induced lollipop. Hence if $G^L_3({\cal I})$ contains an $L$-subdivision, the induced directed cycle in it is the concatenation of the path $bzya$ and a induced directed $(a,b)$-path in $G_3({\cal I})$.
Thus ${\cal I}$ is satisfiable. The other direction is (as usual) clear.
\end{proof}

\begin{remark}
The {\it cone} is  the digraph $C$ with vertex set $\{x,y,z\}$ and arc set  $\{xy, xz, yz, zy\}$.
In the very same way as Theorem~\ref{lollipop}, one can show that finding an induced subdivision of the cone in a digraph
is NP-complete.
\end{remark}

\section{Polynomial-time algorithms for induced subdivisions in oriented graphs}

According to Conjecture~\ref{D-in-dig-conj}, the only digraphs for which $\Pi_D$ is polynomial-time solvable are
disjoint unions of spiders and possibly one  $2$-cycle. For such digraphs, easy polynomial-time algorithms exist (See Section~\ref{sec:EasyPoly}).

In this section, we show that the picture is more complicated for $\Pi'_D$ than for $\Pi_D$.
We show some oriented graphs $D$ for which $\Pi'_D$ is polynomial-time solvable.
For all these oriented graphs, $\Pi_D$ is NP-complete by Corollary~\ref{D-in-dig}.

\subsection{Induced subdivision of cherries in oriented graphs}\label{cherrysec}

Let $s, u, v$ be three vertices such that $s\neq v$ and $u\neq v$ (so
$s=u$ is possible).  A \emph{cherry on $(s, u, v)$} is any oriented
graph made of three induced directed paths $P, Q, R$ such that:
\begin{itemize}
\item $P$ is directed from $s$ to $u$ (so when $s=u$ it has
  length 0);
\item $Q$ and $R$ are both directed from $u$ to $v$ (so they both have
  length at least~1 and since we do not allow parallel edges, at least
  one of them has length at least 2);
\item $u, v$ are the only vertices in more than one of $P, Q, R$;
\item there are no other arcs than those from $P, Q, R$.
\end{itemize}

\noindent The cherry is \emph{rooted at $s$}.  

\vs

An induced cherry contains an induced $TT_3$-subdivision (made of $Q$
and $R$) and a $TT_3$-subdivision is a cherry (with $u=s$).  Hence
detecting an induced cherry is equivalent to detecting an induced
$TT_3$-subdivision.

In order to give an algorithm that detects a cherry rooted at a given
vertex, we use a modification of the well-known Breadth First Search
algorithm (BFS), see e.g.  \cite[Section 3.3]{livre:digraph}.  Given an oriented graph $G$ and a vertex $s\in V(D)$, BFS returns an out-tree rooted
at $s$ and spanning all the vertices reachable from $s$.  It proceeds
as follows:

\begin{algorithm}[H]
\DontPrintSemicolon \Indm {\bf BFS($G,s$)\\}

\Indp
Create a queue $Q$ consiting of $s$; Intialize $T=(\{s\}, \emptyset)$\;
  \While{$Q$ is not empty}{
      Consider the head $u$ of $Q$ and \emph{visit} $u$, that is\;
      \ForEach{out-neighbour $v$ of $u$ in $D$}{
          \If{$v \notin V(T)$}{
               $V(T):=V(T)\cup \{v\}$ and $A(T):=A(T)\cup \{uv\}$\;    
               Put $v$ to the end of $Q$\;
          }
      Delete  $u$ from $Q$\;
      }          
  }
  \end{algorithm}
  
\noindent  
Note that the arc-set of the out-branching produced by BFS depends on the order in which the
vertices are visited, but the vertex-set is always the same: it is the
set of the vertices reachable from $s$. See \cite{livre:digraph} p. 92
for more details on BFS.  We need the following variant:

\begin{algorithm}[H]
\DontPrintSemicolon \Indm {\bf IBFS($G,s$)\\}

\Indp
Create a queue $Q$ consisting of $s$; Intialize $T=(\{s\}, \emptyset)$\;
  \While{$Q$ is not empty}{
      Consider the head $u$ of $Q$ and \emph{visit} $u$, that is\;
      \ForEach{out-neighbour $v$ of $u$ in $G$}{
          \If{$N_G(v)\cap V(T) =\{u\}$}{
               $V(T):=V(T)\cup \{v\}$ and $A(T):=A(T)\cup \{uv\}$\;    
               Put $v$ to the end of $Q$\;
          }
      Delete  $u$ from $Q$\;
      }          
  }
  \end{algorithm}

  Observe that IBFS (which we also call {\em induced-BFS}) is the same
  as BFS except that we add the out-neighbour $v$ of $u$ to $T$ only if
  it has no other neighbour already in $T$, hence ensuring that the
  resulting out-tree is an induced subdigraph of $G$.  Contrary to
  BFS, the vertex-set of a tree obtained after IBFS
  may depend on the order in which the vertices are visited.

  IBFS can easily be implemented to run in time $O(n^2)$.  When $T$ is
  an oriented tree, we denote by $T[x,y]$ the unique oriented path from $x$ to $y$ in $T$.

\begin{theorem}\label{keyIBFS}
  Let $G$ be an oriented graph, $s$ a vertex and $T$ a tree obtained
  after running $IBFS(G,s)$.  Then exactly one of the following
  outcomes is true:
  \begin{enumerate}
  \item\label{i:ch} $D$ contains an induced subdigraph that is a cherry rooted at
    $s$; 
  \item\label{i:obstr} for every vertex $x$ of $T$, any
    out-neighbour of $x$ not in $T$ has an out-neighbour that is an
    ancestor of $x$ in $T$.
  \end{enumerate} 

  This is algorithmic in the sense that there is an $O(n^2)$ algorithm
  that either outputs the cherry of \ref{i:ch} or checks that
  \ref{i:obstr} holds. 
\end{theorem}

\begin{proof}
  Suppose that $T$ does not satisify \ref{i:obstr}.  Then some vertex
  $x$ of $T$ has an out-neighbour $y$ not in $T$ and no out-neighbour
  of $y$ is an ancestor of $x$.  Without loss of generality, we 
  assume that $x$ is the first vertex added to $T$ when running IBFS
  with such a property.  In particular, $T[s,x]y$ is an induced directed
  path because a chord would contradict \ref{i:obstr} or the choice of
  $x$.  Let $v$ be the neighbour of $y$ in $T$, different from $x$,
  that was first added to $T$ when running IBFS.  Note that $v$ exists
  for otherwise $y$ would have been added to $T$ when visiting $x$.
  If $x$ is the parent of $v$ in $T$ then $T[s,x]y$ together with $v$
  form a cherry rooted at $s$ (whatever the orientation of the arc
  between $y$ and $v$).  So we may assume that $x$ is not the parent
  of $v$.  When visiting $x$, vertex $y$ was not added to $T$, hence
  $v$ was already visited (because $x$ is not the parent of $v$).  In
  addition, when $v$ was visited, it was the unique neighbour of $y$
  in the current out-tree, so $y$ is an in-neighbour of $v$, for
  otherwise it would have been added to $T$.  Let $u$ be the common
  ancestor of $x$ and $v$ in $T$, chosen closest to $x$.  Since $T$
  does not satisfy \ref{i:obstr} by the choice of $x$ and $y$, $u\neq v$.  Now the directed paths
  $sTu$, $T[u,x]yv$ and $T[u,v]$ form an induced cherry rooted at $s$.
  Indeed since $T$ is an induced out-tree, it suffices to prove that
  $y$ has no neighbour in these three paths except $x$ and $v$. By
  definition of $v$, there is no neighbour of $y$ in $T[s,u]$ and $T[u,v]$
  except $v$.  Moreover, $y$ has no out-neighbour in $T[u,x]$ by the 
  assumption that (ii) does not hold for $y$ and $x$ and it has 
no in-neighbour in $T[u,x]$ except $x$ by the choice of
  $x$.

  Conversely, let us assume that $T$ satisfies~\ref{i:obstr} and
  suppose by contradiction that $G$ contains an induced cherry $C$
  rooted at $s$.  Since $T$ is an induced out-branching, some vertices
  of $C$ are not in $T$.  So, let $y$ be a vertex of $V(C)\setminus
  V(T)$ as close to $s$ as possible in the cherry.  Let $x$ be an
  in-neighbour of $y$ in $C\cup T$.  From the choice of $y$, $x$ and
  all its ancestors along the cherry are in $T$.  Since $T$ is
  induced, the ancestors of $x$ along the cherry are in fact the
  ancestors of $x$ along $T$.  Hence, $x$ is a vertex of $T$ with an
  out-neighbour $y$ not in $T$ having no out-neighbour among the
  ancestors of $x$ along $T$.  This contradicts $T$
  satisifying~\ref{i:obstr}.

  All this may be turned in an $O(n^2)$-algorithm that finds a cherry
  rooted at $s$ if it exists or answer no otherwise. Indeed we first
  run IBFS and then check in time $O(n^2)$ if the obtained tree $T$
  satisfies~\ref{i:obstr}. If not, then we can find the cherry
  following the first paragraph of the proof.
\end{proof}

\begin{remark}
  Since a digraph contains an induced $TT_3$-subdivision if and only
  if it contains an induced cherry, Theorem~\ref{keyIBFS} implies
  directly that $\Pi'_{TT_3}$ is solvable in time $O(n^3)$ (because we
  need to enumerate all potential roots).
\end{remark}

We can slightly extend our result.  A {\it tiny cherry} is a cherry
such that the path $Q$ and $R$ as in the definition form a $TT_3$.

\begin{corollary}
  For any tiny cherry $D$, the problem $\Pi'_D$ is solvable in time
  $O(n^{|V(D)|})$.
\end{corollary}
\begin{proof}
  Let $P$ be the path of $D$ as in the definition of cherry.  Let $G$
  be the input oriented graph.  Enumerate by brute force all induced
  directed paths of order $|P|$ by checking all the possible
  subdigraphs of order $|P|$.  For each such path $P'$ with terminus
  $x$, look for a cherry rooted at $x$ in the graph $G'$ obtained by
  deleting all the vertices of $P-x$ and their neighbourhoods except
  $x$.  If there is such a cherry $C$ then the union of $P$ and $C$ is
  an induced $D$-subdivision.
 \end{proof}

Similarly to Propositions~\ref{easy2} and \ref{easy3}, we have the following.
\begin{corollary}
If $D$ is the disjoint union of spiders and a tiny cherry then $\Pi'_D$ is polynomial-time solvable.
\end{corollary}

\subsection{Induced subdivision of oriented paths with few blocks in oriented graphs.}\label{subsec:or-path}

By Proposition~\ref{easy}, for any oriented path $P$ with at most two blocks $\Pi_P$ and thus $\Pi'_P$ are polynomial-time solvable.
In this section, we shall prove that $\Pi'_P$ is polynomial-time solvable for some oriented paths with three or four blocks. In contrast, $\Pi_P$ is NP-complete for every oriented path with at least three blocks as shown in Corollary~\ref{D-in-dig}.

\subsubsection{Oriented path with three blocks}

\begin{theorem}\label{A2-origin}
  There exists an algorithm of complexity $O(m^2)$ that given a connected
  oriented graph on $n$ vertices and $m$ arcs with a specified vertex $s$
  returns an induced $A^+_2$-subdivision with origin $s$ if one
  exists, and answer `no' if not.
\end{theorem}
\begin{proof}
  Observe that any induced $A^+_2$-subdivision with origin $s$
  contains an induced $A^+_2$-subdivision with origin $s$ 
 such that the directed path corresponding to the arc $s_3s_2$ is
  some arc $f$. Such a subdivision is called {\it $f$-leaded}.

  Given an oriented graph $G$, we enumerate all arcs $f=s'_3s'_2$.  For each
  arc in turn we either show that there is no $f$-leaded induced
  $A^+_2$-subdivision with origin $s$ or give an induced subdivision of
  $A^+_2$ with origin $s$, (but not necessarily $f$-leaded).  This will
  detect the $A^+_2$-subdivision since if some exists, it is $f$-leaded
  for some $f$. 
  
    We do this as follows.  We delete all
  in-neighbours of $s$ and all neighbours of $s'_3$  except $s'_2$. 
  Let us denote by $G'$ the resulting graph.
  Then we compute by BFS a shortest directed path $P$ from $s$ to $s'_2$.
  If it is induced, together with $s'_3s'_2$, it forms the desired
  $A^+_2$ subdivision.  So, as $P$ has no forward chord (since it is a
  shortest path), there is an arc $uv$ in $G'\langle V(P)\rangle$ such
  that $u$ occurs after $v$ on $P$. Take such an arc $b_3b_2$ such
  that $b_2$ is as close as possible to $s$ (in $P$).  Observe that
  since we deleted all in-neighbours of $s$ we have $b_2\neq s$.
  Now, $P[s_1,  b_2]$ together with $b_3b_2$ forms the desired
  $A^+_2$-subdivision.

  \vspace{6pt} There are $O(m)$ arcs and for each of them we must find
  a shortest path in $G'$ which can be done in $O(m)$. Hence the
  complexity of the algorithm is $O(m^2)$.
\end{proof}

From this theorem, one can show that finding an induced $A^-_3$-subdivision is polynomial-time solvable.
It is enough to enumerate all arcs $s'_2s'_1$, to delete $s'_1$ and its
 neighbours except $s'_2$, and to decide whether there exists in what
 remains an $A^+_2$-subdivision with origin $s_2$.  
 One can also derive polynomial-time algorithms for finding induced subdivisions of other oriented paths with three blocks.

\begin{corollary}
  Let $P$ be a path with three blocks such that the last one has
  length $1$.  One can check in time $O(n^{|P|-2}m^2)$ whether a given
  oriented graph contains an induced $P$-subdivision.
\end{corollary}
\begin{proof}
  By directional duality, we may assume that $P$ is an
  $A^-_3$-subdivision.  Let  $Q$ be the subdigraph of $P$ formed
  by the first block of $P$ and the second block of $P$ minus the last
  arc.  Let $s$ be the terminus of $Q$.  For each induced oriented
  path $Q'$ in the instance graph, isomorphic to $Q$ (there are at
  most $O(n^{|P|-2})$ of them), we delete $Q'-s$ and all vertices that
  have neighbours in $Q-s$ except $s$.  We then detect an
  $A^+_2$-subdivision rooted at~$s$ in the resulting graph.  This will
  detect a $P$-subdivision if there is one.
\end{proof}

\subsubsection{Induced subdivision of $A^-_4$ in an oriented graph}\label{sec:A4}

We show how to check the presence of an induced copy of $A^-_4$ by using
flows (for definitions and algorithms for flows see e.g. \cite[Chapter
4]{livre:digraph}). 

\begin{theorem}\label{mcflowsA4}
There exists an algorithm of complexity $O(nm^2)$ that given an
oriented graph on $n$ vertices and $m$ arcs with a specified vertex
$s$ returns an induced $A^+_3$-subdivision rooted at $s$, if one exists,
and answer `no' if not.
\end{theorem}
\begin{proof}
  The general idea is close to the one of the proof of
  Theorem~\ref{A2-origin}.  Observe that any induced $A^+_3$-subdivision with origin $s=a_1$ contains an induced subdivision of $A^+_3$ with origin $s=a_1$
  such that the directed path corresponding to the arc $s_3s_4$ is
  some arc $f$. If, in addition, the vertex corresponding to $s_2$ is
  $v$, such a subdivision is called {\it $(v,f)$-leaded}.

  Given an oriented graph $G$, we enumerate all pairs $( a_2,
  a_3a_4)$ such that $a_2, a_3, a_4$ are distinct vertices and
  $a_3a_4\in E(G)$.  For each such pair in turn we either show that
  there is no $(a_2, a_3a_4)$-leaded induced $A^+_3$-subdivision with origin $a_1$ or give an induced subdivision of $A^+_3$ with origin $a_1$ (but not necessarily $(a_2, a_3a_4)$-leaded).

  We do this as follows.  We first delete all the neighbours of $a_4$
  except $a_3$, all in-neighbours of $a_1$ and $a_3$ and finally all
  out-neighbours of $a_2$. If this results in one or more of the
  vertices $a_1, \dots , a_4$ to be deleted, then there cannot be any
  $(a_2, a_4a_3)$-leaded induced $A^+_3$-subdivision with origin $a_1$
  because there is an arc in $G\langle\{ a_1, \dots , a_4\}\rangle$
  which is not in $\{a_1a_2, a_3a_2, a_3a_4\}$.  So we skip this
  pair and proceed to the next one.  Otherwise we delete $a_4$ and
  we use a flow algorithm to check in the resulting digraph $G'$ the
  existence of two internally-disjoint directed paths $P,Q$ such that
  the origin of $P$ and $Q$ are $a_1$ and $a_3$ respectively and such that $a_2$ is the terminus of both
    $P$ and $Q$.
  Moreover, we suppose that these two paths have no forward chord
  (this can easily be ensured by running BFS on the graphs induced by
  each of them).  If no such paths exist , then we proceed to the next pair
  because there is no $(a_2, a_3a_4)$-leaded induced
  $A^+_3$-subdivision.  If we find such a pair of directed paths $P,Q$,
  then we shall provide an induced subdivision of $A^+_3$ with origin
  $a_1$.  If $P$ and $Q$ are induced and have no arcs between them,
  then these paths together with the arc $a_3a_4$ form the desired
  induced subdivision of $A^+_3$.

  Suppose that $P$ is not induced.  As $P$ has no forward chord, there
  is an arc $uv$ in $G'\langle V(P)\rangle$ such that $u$ occurs after
  $v$ on $P$. Take such an arc $b_3b_2$ such that $b_2$ is as close as
  possible to $a_1$ (in $P$), and subject to this, such that $b_3$ is
  as close as possible to $a_2$. Observe that since we deleted all
  in-neighbours of $a_1$ and all out-neighbours of $a_2$ before, we
  must have $b_2\neq a_1$ and $b_3\neq a_2$. Let $b_4$ be the
  successor of $b_3$ on $P$.  Now $P[a_1, b_2]$ and the arcs
  $b_3b_2,b_3b_4$ form the desired induced subdivision of $A^+_3$.
  From here on, we suppose that $P$ is induced. 

  Suppose now that there is an arc $e$ with an end $x\in V(P)$ and the
  other $y\in V(Q)$. Choose such an arc so that the sum of the lengths
  of $P[a_1, x]$ and $Q[a_3, y]$ is as small as possible.  If $e$ is from $x$
  to $y$ we have $y\neq a_3$ because we removed all the in-neighbours
  of $a_3$, else $e$ is from $y$ to $x$ and we have $x\neq a_1$
  because we removed all the in-neighbours of $a_1$.  In all cases,
  we get an induced subdivision of $A^+_3$ by taking the paths
  $P[a_1, x]$ and $Q[a_3, y]$ and the arcs $a_3a_4, e$.  From here on, we
  suppose that there are no arcs with an end in $V(P)$ and the
  other in $V(Q)$.

  The last case is when $Q$ is not induced.  Since $Q$ has no forward
  chord, there is an arc $uv$ in $G'\langle V(Q)\rangle$ such that $u$
  occurs after $v$ on $Q$. Take such an arc $b_3b_4$ such that $b_3$
  is as close as possible to $a_2$ (in $Q$).  Observe that since we
  deleted all out-neighbours of $a_2$ before, we must have $b_3\neq
  a_2$.  Now $P$, $Q[b_3,  a_2]$ and the arc $b_3b_4$ form the desired
  induced subdivision of $A^+_3$.

   \vspace{6pt}
  
   There are $O(nm)$ pairs $(a_2, a_3a_4)$ and for each of
   them, we run an $O(m)$  flow algorithm (we just need to find a flow of value 2, say, by the Ford-Fulkerson method \cite[Section 4.5.1]{livre:digraph}) and do some
   linear-time operations.  Hence the complexity of the algorithm is
   ${\mathcal O}(nm^2)$.
\end{proof}

  One can check in polynomial time if there is an induced $A^-_4$-subdivision: it is
  enough to enumerate all arcs $t_2t_1$, to delete $t_1$ and its
  neighbours except $t_2$, and to decide whether there exists in what
  remains an $A^+_3$ subdivision with origin $t_2$.  
One can also derive polynomial-time algorithm for finding induced subdivision of other oriented paths with four blocks.

\begin{corollary}
  Let $P$ be an oriented path that can be obtained from $A^-_4$ by
  subdividing the first arc and the second arc.  One can check in
  time  $O(n^{|P|-1}m^2)$ whether a given oriented graph contains an induced
  subdivision of $P$.
\end{corollary}
\begin{proof} 
  Let $R$ be the subdigraph of $P$  formed by the first block of $P$ and its
  second block minus the last arc.  Let $s$ be the last vertex of $R$.
  For each induced oriented path $Q$ in the instance graph, isomorphic
  to $R$ (there are $O(n^{|P|-3})$ of them), we delete $Q-s$, all
  vertices that have neighbours in $Q-s$ except $s$ and detect an
  $A^-_3$-subdivision with origin~$s$.  This will detect a
  $P$-subdivision if there is one.
\end{proof}

\section{Induced subdivisions of tournaments in oriented graphs}\label{sec:NP-tour}

\subsection{Induced subdivision of transitive tournaments}

The transitive tournament on $k$  vertices is denoted $TT_k$. We saw in Section \ref{cherrysec} that $\Pi'_{TT_3}$ is polynomial. The next result shows that $\Pi'_{TT_k}$ is NP-complete for all $k\geq 4$.

\begin{figure}
\begin{center}
\scalebox{0.6}{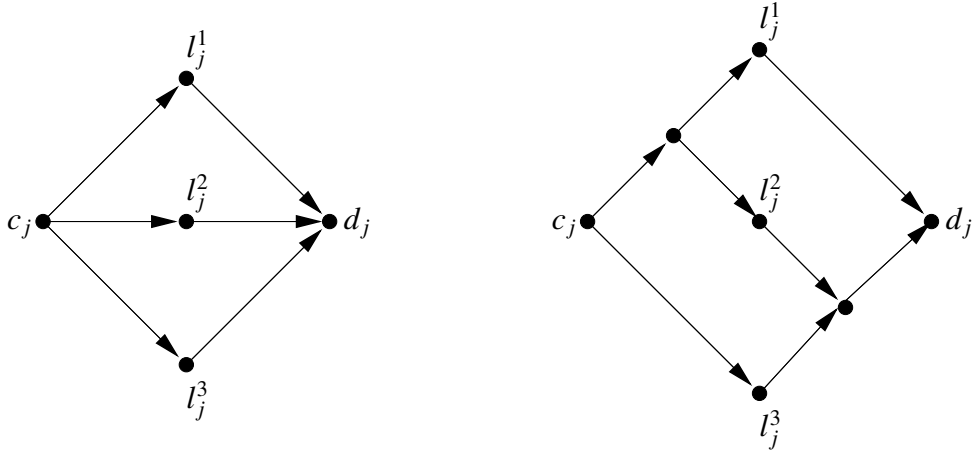}
\end{center}
\caption{Left: clause gadget of $G_1({\cal I})$. Right: clause gadget of $G_4({\cal I})$.}\label{modifiedgadget}
\end{figure}

\begin{theorem}
For all $k\geq 4$, $\Pi'_{TT_k}$ is NP-complete
\end{theorem}

\begin{proof}
  For a given instance $\cal I$ of 3-SAT, let $G_4({\cal I})$ be the
  digraph we obtain from $G_1({\cal I})$ by replacing each clause
  gadget $C^1_j$ by the modified one $C^4_j$ from Figure
  \ref{modifiedgadget}.  Also for each variable, modify the gadget
  $V_i^1$ as follows: replace the path $a_i x_i v_i b_i$ by a path
  $a_i x_i^1 v_i^1 x_i^2 v_i^2 \dots x_i^m v_i^m b_i$, and similarly
  the path $a_i \bar x_i \bar v_i b_i$ by a path $a_i \bar x_i^1 \bar
  v_i^1 \bar x_i^2 \bar v_i^2 \dots \bar x_i^m \bar v_i^m b_i$.  Then
  in $G_4(I)$ the links representing a variable $x_i$  and a clause $C_j$ that uses this
  variable are represented by arcs between vertices from the variable
  gadget with superscript $j$ (as in Figure~\ref{inducedpathfig}).

  Recall that $G_1({\cal I})$ has an induced
  directed $(a,b)$-path if and only if $\cal I$ is satisfiable. It is
  easy to see that the same holds for $G_4({\cal I})$. Note that in
  $G_4({\cal I})$ no vertex has in- or out-degree larger than 2.

Given an instance $\cal I$ of 3-SAT we form the digraph $G^k_4({\cal I})$ from $G_4({\cal I})$ and a copy of $TT_k$ (with vertices $v_1,v_2,\dots , v_k$ and arcs $v_iv_j$, $1\leq i<j\leq k$) by deleting the arc $v_1v_k$ and adding the arcs $v_1a,bv_k$. We claim that $G^k_4({\cal I})$ contains an induced subdivision  of $TT_k$ if and only if $G_4({\cal I})$ has an induced directed $(a,b)$-path  which is if and only if $\cal I$ is satisfiable. 

Clearly,  if $\cal I$ is satisfiable, we may use the concatenation of an induced directed $(a,b)$-path in $G_4({\cal I})$ with $v_1a$ and $bv_k$ in place of
$v_1v_k$ to obtain an induced $TT_4$-subdivision in $G^k_4({\cal I})$. 

Conversely, suppose that $G^k_4({\cal I})$ contains an induced
subdivision of $TT_k$ and let $h(v_i)$, $1\leq i \leq k$, denote the image
of $v_i$ in some fixed induced subdivision $H$ of $TT_k$. Then we must
have $h(v_1)=v_1$ and $h(v_k)=v_k$, because $G_4({\cal I})$ does not
contain any vertex of out-degree $k-1$ or in-degree $k-1$ because $k\geq 4$. 
For all $i$, $1<i<k$, the vertex $h(v_i)$ could not be in $V(G_4({\cal I}))$ since otherwise there must be either two disjoint directed 
$(v_i,v_k)$-paths to $v_k$ or two disjoint directed $(v_1,v_i)$-paths. This is impossible because there is no directed $(v_i,v_k)$-path
 in $G^k_4({\cal I})\setminus bv_k$ and no directed $(v_1,v_i)$-path in $G^k_4({\cal I})\setminus v_1a$.
Hence $h(v_i)=v_i$ for all $1\leq i \leq k$ and so it is clear that
we have an induced directed $(a,b)$-path in $G_4({\cal I})$, implying
that $\cal I$ is satisfiable.
\end{proof}

In the proof above we used that the two vertices $v_1,v_k$ cannot be mapped to vertices of $G_4({\cal I})$, the fact that the connectivity between these and the other vertices is too high to allow any of these to be mapped to vertices of $G_4({\cal I})$ and finally we could appeal 
to the fact that  $G_4({\cal I})$ has an induced directed $(a,b)$-path if and only if $\cal I$ is satisfiable. Refining this argument it is not difficult to see that the following holds where a {\it $(z,X)$-path} is a path whose initial vertex is $z$ and whose last vertex belongs to $X$.

\begin{theorem}
\label{reachability}
Let $D=(V,A)$ be a digraph and let $X$ (resp. $Y$) be the subset of
vertices with out-degree (resp.\ in-degree) at least 3 and let
$Z=V\setminus (X\cup Y)$ (note that $X\cap Y\neq\emptyset$ is possible
and also $Z=\emptyset$ is possible). Suppose that for every $z\in Z$
the digraph $D$ contains either two internally disjoint $(X,z)$-paths
or two internally disjoint $(z,Y)$-paths. Then $\Pi'_D$ is
NP-complete.
\end{theorem}
\begin{proof} (Sketch) For a given instance $\cal I$ of 3-SAT we form the digraph $G'_4({\cal I})$ from $D$ by replacing one arc $uv$ with at least one of its endvertices in $X\cup Y$ by $G_4({\cal I})$ and the arcs $ua,bv$. Again it is clear how to obtain an induced subdivision of $D$ in $G'_4({\cal I})$ when $\cal I$ is satisfiable.
Let us now assume that $G'_4({\cal I})$ contains an induced subdivision $D'$ of $D$. Let $\{h(v)|v\in V\}$ be the vertices corresponding to the vertices of $V$ in the subdivision.
For degree reasons, none of the vertices in $X\cup Y$ can have $h(v)\in V(G_4({\cal I}))$ and because of connectivity, none of the vertices of $Z$ can have $h(z)\in V(G_4({\cal I}))$ because there is only one arc entering and leaving $V(G_4({\cal I}))$ in $G'_4({\cal I})$.
Thus $\{h(v)|v\in V\}=V$ (possibly with $h(v)\neq v$ for several vertices). However, since we deleted the arc $uv$ and replaced it by $G_4({\cal I})$ and the arcs $ua,bv$, it follows that $G'_4({\cal I})$ and so $G_4({\cal I})$ contains an induced directed $(a,b)$-path, implying that $\cal I$ is satisfiable.
\end{proof}

\subsection{Induced subdivision of the strong tournament on $4$ vertices}

Let $ST_4$ be the unique strong tournament of order $4$.
It can be seen has a directed cycle $\alpha\gamma\beta\delta\alpha$ together with two {\it chords} $\alpha\beta$ and $\gamma\delta$.
The aim of this section is to show that $\Pi'_{ST_4}$ is NP-complete.

An {\it $(x, y_1, y_2)$-switch} is the digraph with vertex set $\{x, z, y_1, y_2\}$ and edge set
$\{xz, xy_1, zy_1, zy_2, y_2y_1\}$.
See Figure~\ref{switch-fig}.

\begin{figure}[hbtp]
\begin{center}
\scalebox{0.5}{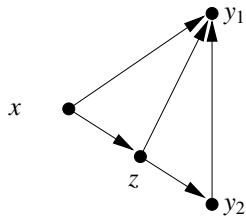}
\end{center}
\caption{An $(x, y_1, y_2)$-switch.}
\label{switch-fig}
\end{figure}
A {\it good $(x, y_1, y_2)$-switch} in a digraph $D$ is an induced switch $Y$ such that
all the arcs entering $Y$ have head $x$ and all arcs leaving $Y$ have tail in $\{y_1,y_2\}$.

\begin{lemma}\label{switch}
Let $Y$ be a good $(x,y_1,y_2)$-switch in a digraph $D$.
Then every induced subdivision $S$ of $ST_4$ in $D$ intersects $Y$ on either
the path $(x,y_1)$, the path $(x,z,y_2)$, or the empty set.
\end{lemma}
\begin{proof}
Suppose for a contradiction, that $y_2y_1\in E(S)$. Then $S$ must contain the unique in-neighbour $z$ of $y_2$
and the unique in-neighbour $x$ of $z$. Hence $y_1$ has in-degree $3$ in $S$, a contradiction.

Suppose for a contradiction, that $zy_1\in E(S)$. Then $S$ must contain $x$ the unique in-neighbour of $z$. 
Hence $xy_1$ is a chord of $S$ and so $z$ must have degree $3$ in $S$.
Thus $y_2 \in V(S)$ and $y_1$ has in-degree $3$ in $S$, a contradiction.
\end{proof}

\begin{theorem}\label{ST4}
$\Pi'_{ST_4}$ is NP-complete.
\end{theorem}
\begin{proof}
Reduction from $3$-SAT.
Let ${\cal I}$ be an instance of $3$-SAT with variables
  $x_1,x_2,\ldots{},x_n$ and clauses $C_1,C_2,\dots{},C_m$. We first
  create a variable gadget $V^5_i$ for each variable $x_i$,
  $i=1,2,\ldots{},n$ and a clause gadget $C^5_j$ for each clause $C_j$,
  $j=1,2,\ldots{},m$ as shown in Figure \ref{G4fig}. Then we form
  the digraph $G_5({\cal I})$ as follows: Form a chain $U$ of variable gadgets by
  adding the arcs $b_ia_{i+1}$ for $i=1,2,\ldots{},n-1$ and a chain
  $W$ of clause gadgets by adding the arcs $d_jc_{j+1}$,
  $j=1,2,\ldots{},m-1$. Add the arcs $aa_1,b_nb, cc_1,t_md$.  For each clause $C$, we connect the three literal
  vertices of the gadget for $C$ to the variable gadgets for variables
  occuring as literals in $C$ in the way indicated in Figure~\ref{G4connect}.

\begin{figure}[hbtp]
\begin{center}
\scalebox{0.5}{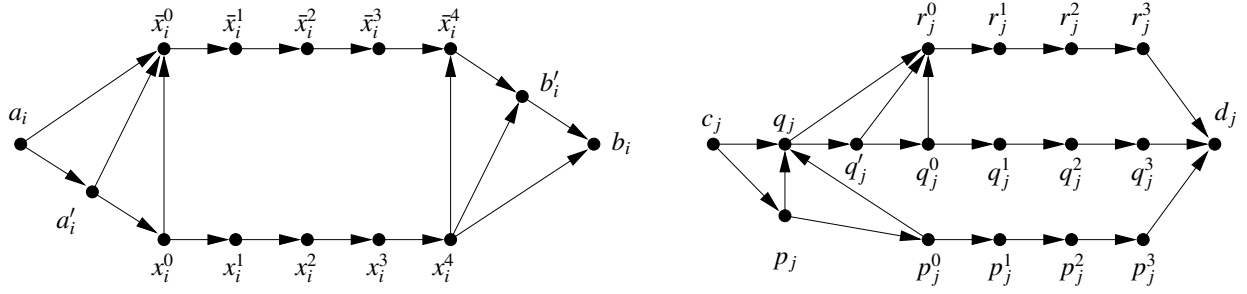}
\caption{The variable gadget $V^5_i$ (left) and the clause gadget $C^5_j$ (right).}
\label{G4fig}
\end{center}
\end{figure}

\begin{figure}[hbtp]
\begin{center}
\scalebox{0.5}{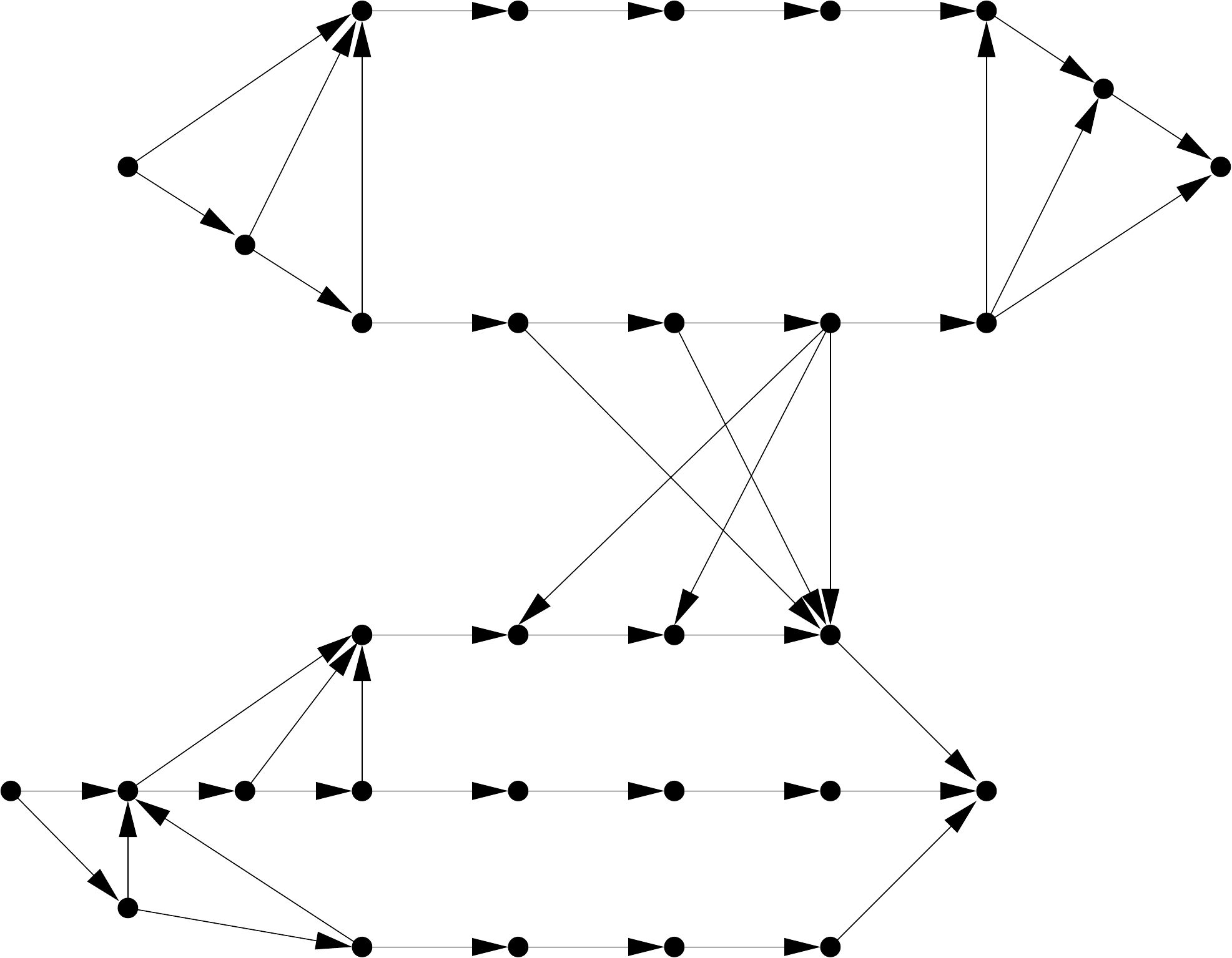}
\caption{Connections between a clause gadget and a variable gadget in $G_5({\cal I})$. Only the connection for one variable gadget and one clause gadget is shown and the general strategy for connecting variable and clause gadgets is the same as in $G_1(I)$ (Figure \ref{inducedpathfig}).}
\label{G4connect}
\end{center}
\end{figure}

We denote by $X_i$ the path  $a_ia'_ix^0_ix^1_ix^2_ix^3_ix^4_ib_i$, by
$\bar{X}_i$ the path  $a_i\bar{x}^0_i\bar{x}^1_i\bar{x}^2_i\bar{x}^3_i\bar{x}^4_ib'_ib_i$,
by $P_j$ the path $c_jp_jp^0_jp^1_jp^2_jp^3_jd_j$,
by $Q_j$ the path $c_jq_jq'_jq^0_jq^1_jq^2_jq^3_jd_j$,
and by $R_j$ the path $c_jq_jr^0_jr^1_jr^2_jr^3_jd_j$.

Similarly to the proof of Lemma~\ref{inducedpath}, one can check that 
${\cal I}$ is satisfiable if and only if there are two induced disjoint directed $(a,b)$- and $(c,d)$-paths  in $G_5({\cal I})$.

Let $G^*_5({\cal I})$ be the digraph obtained from $G_5({\cal I})$ by adding the edges $ac$, $cb$, $bd$ and $da$.
Observe that $G^*_5({\cal I})\setminus da$ is acyclic.

Let us prove that $G^*_5({\cal I})$ contains an induced $ST_4$-subdivision if and only if ${\cal I}$ is satisfiable.  

If $\cal I$ is satisfiable, then in $G_5({\cal I})$ there are two induced disjoint directed $(a,b)$- and $(c,d)$-paths.
The union of these paths and the directed cycle $acbd$ is an induced $ST_4$-subdivision in $G^*_5({\cal I})$.

\vspace{6pt}

Conversely, assume that $G^*_5({\cal I})$ contains an induced subdivision $S$ of $ST_4$.
For sake of simplicity (and with a slight abuse of notation), we will denote the vertices of $S$
corresponding to $\alpha$, $\beta$, $\gamma$ and $\delta$ by the same names.
Let $T_1$ and $T_2$ be the paths corresponding to the chord $\alpha\beta$ and $\gamma\delta$ respectively in $S$ and let $C$ be the directed cycle corresponding to $\alpha\gamma\beta\delta\alpha$. Observe that the ends of $T_1$ and $T_2$ must alternate on $C$.

Notice that the subdigraphs induced by the sets 
$\{a_i, a'_i, x_i^0, \bar{x}_i^0\}$, $1\leq i\leq n$, 
$\{c_j,p_j,p_j^0, q_j\}$ and $\{q_j,q'_j,q_j^0, r_j^0\}$ are good switches.
In addition, the subdigraphs induced by the sets 
$\{b_i, b'_i, x_i^4, \bar{x}_i^4\}$, $1\leq i\leq n$, are the converse of good switches.
Hence Lemma~\ref{switch} (and its converse) imply the following proposition.

\begin{claim}\label{novert1} ~
\begin{itemize}
\item[(i)] For $1\leq  i\leq n$,  if $a_i\in V(S)$, then exactly one of the two paths
$(a_i,a'_i, x_i^0)$ and  $(a_i, \bar{x}_i^0)$ is in $S$.

\item[(ii)] For $1\leq  i\leq n$,  if $b_i\in V(S)$, then exactly one of the two paths
$(\bar{x}_i^4,b'_i, b_i)$ and  $(x^4_i, b_i)$ is in $S$.

\item[(iii)] For $1\leq j\leq m$, if $c_j\in V(S)$, then exactly one of the three paths
$(c_j,p_j, p_j^0)$, $(c_j,q_j, q'_j, q_j^0)$  and  $(c_j, q_j, r_j^0)$ is in $S$.
\end{itemize}
\end{claim}

Since $G^*_5({\cal I})\setminus da$ is acyclic,  $C$ must contain the arc $da$.
Moreover since there is no arc with tail in some clause gadget and head in some variable gadget, $C$ contains at most one arc with tail in some variable gadget and head in some clause gadget.

\begin{claim}\label{nodown1}
For any $1\leq i\leq n$ and any $1\leq j\leq m$, the cycle $C$ contains no arc with tail in $\{x^3_i,\bar{x}_i^3\}$ and  head  in $\{p_j^1, q_j^1,r_j^1\}$.
\end{claim}
\begin{proof}
Assume for a contradiction that  $C$ contains such an arc $y^3_il_j^1$. Then since $l_j^1$ and $l_j^2$ have out-degree $1$ then $C$ must also contain $l_j^2$ and $l_j^3$. Thus, in $S$, $y^3_i$ has out-degree $3$ in $S$, a contradiction.
\end{proof}

\begin{claim}\label{nodown1-bis}
For any $1\leq i\leq n$ and any $1\leq j\leq m$ the cycle $C$ contains no arc with tail in $\{x^3_i,\bar{x}_i^3\}$ and  head  in $\{p_j^3, q_j^3,r_j^3\}$.
\end{claim}
\begin{proof}
Assume for a contradiction that  $C$ contains such an arc $y^3_il_j^3$. Then since $y^3_i$ and $y^2_i$ have in-degree $1$ then $C$ must also contain $y^2_i$ and $y^1_i$. Thus, in $S$, $l^3_j$ has in-degree $3$ in $S$, a contradiction.
\end{proof}

\begin{claim}\label{nodown2}
For any $1\leq i\leq n$ and any $1\leq j\leq m$ then $C$ contains no arc with tail in $\{x_i^3,\bar{x}_i^3\}$ and  head  in $\{p_j^2, q_j^2,r_j^2\}$.
\end{claim}
\begin{proof}
Assume for a contradiction that  $C$ contains such an arc $y^3_il_j^2$. The vertex $l_j^2$ has a unique out-neighbour  $l_j^3$ which  must be in $C$. It follows that $y^3_il_j^3$ corresponds to one of the chords $\alpha\beta$ or $\gamma\delta$.
Thus $l_j^2$ must have degree $3$ in $S$. It follows that $l_j^1$ is in $V(S)$ and so $y^3_i$ has out-degree $3$ in $S$, a contradiction.
\end{proof}

\begin{claim}\label{nodown2-bis}
For any $1\leq i\leq n$ and any $1\leq j\leq m$ the cycle $C$ contains no arc with tail in $\{x_i^2,\bar{x}_i^2\}$ and  head  in $\{p_j^3, q_j^3,r_j^3\}$.
\end{claim}
\begin{proof}
Assume for a contradiction that  $C$ contains such an arc $y^2_il_j^3$. The vertex $y_i^2$ has a unique in-neighbour  $y_i^1$ which  must be in $C$. It follows that $y^1_il_j^3$ corresponds to one of the chords $\alpha\beta$ or $\gamma\delta$.
Thus $y_i^2$ must have degree $3$ in $S$. It follows that $y_i^3$ is in $V(S)$ and so $l_j^3$ has in-degree $3$ in $S$, a contradiction.
\end{proof}

\begin{claim}\label{nodown3}
For any $1\leq i\leq n$ and any $1\leq j\leq m$ the cycle $C$ contains no arc with tail in $\{x^1_i,\bar{x}^1_i\}$ and  head  in $\{p_j^3, q_j^3,r_j^3\}$.
\end{claim}
\begin{proof}
Assume for a contradiction that  $C$ contains such an arc $y^1_il_j^3$. Without loss of generality $y^1_i=x^1_i$.

By the remark after Claim \ref{novert1} this is the only arc from a variable gadget to a clause gadget. Furthermore, we have that $b$ is not on $C$.

Thus, by Claim~\ref{novert1}, for every $1\leq k<i$, the intersection of $C$ and $V^4_{k}$ is either $X_{k}$ or $\bar{X}_{k}$, and
for every $j<l\leq m$, the intersection of $C$ and $C^5_{j}$ is either $P_{j}$,  $Q_j$ or $R _{j}$.

Consider $y\in \{\alpha, \beta\}$. It is on $C$ and has outdegree $2$. On the other hand, applying Claim \ref{novert1} we see that the following must hold
as none of these vertices can belong to $S$ and at the same time have two of their out-neighbours in $S$:

\begin{itemize}
\item $y\not\in \cup_{1\leq j\leq m}\{c_j,p_j,q_j,q'_j,q^0_j\}$,
\item $y\not\in \cup_{k\neq i}\{a_k,a'_k,x^0_k,x^4_k\}$ and
\item $y\not\in \{a_i,a'_i,x^0_i,x^4_i\}$.
\end{itemize}

By Claims \ref{nodown1}-\ref{nodown2-bis}, we have $y\not\in \{x^2_i,x^3_i\}$ and since $b$ is not on $C$ we also have $y\neq b$. If $y=x^1_i$, then using that $yl_j^3$ is and arc of  $C$ we get a contradiction because $x^2_il_j^3$ is an arc (so we cannot obtain an induced copy of $S$ using both arcs $yl_j^3,x^2_il_j^3$).
Hence (as $y$ was any of $\alpha,\beta$) we have  $a=\alpha=\beta$, a contradiction.
\end{proof}

\begin{claim}\label{claim}
$C=acbda$.
\end{claim}
\begin{proof}
Suppose not. Then by the above claims, $C$ either does not intersect the clause gadget and intersect all the
variable ones or does not intersect the variable gadget and intersect all the
gadget ones.
In both cases, similarly to the proof of Claim~\ref{nodown2-bis}, one shows that $a=\alpha=\beta$, a contradiction.
\end{proof}

Since $C=acbda$ and by construction of $G^*_5({\cal I})$, $T_1$ and
$T_2$ are two induced disjoint path in $G_5({\cal I})$ and so ${\cal
  I}$ is satisfiable.
\end{proof}

\section{Remarks and open problems}\label{sec:remarks}

It would be nice to have results proving a full dichotomy between the
digraphs $D$ for which $\Pi_D$ (resp. $\Pi'_D$) is $NP$-complete and
the ones for which it is polynomial-time solvable.  Regarding $\Pi_D$,
Conjecture~\ref{D-in-dig-conj} gives us what the dichotomy should
be. But for $\Pi'_D$ we do not know yet.

A useful tool to prove such a dichotomy would be the following conjecture.
\begin{conjecture}
If $D$ is a digraph such that $\Pi_D$ (resp. $\Pi'_D$) is
 NP-complete, then for any digraph $D'$ that contains $D$ as an
  induced subdigraph,  $\Pi_{D'}$ (resp. $\Pi'_{D'}$) is NP-complete.
\end{conjecture}

We were able to settle the complexity of $\Pi'_D$ when $D$ is a directed
cycle, a directed path,  or some paths with at most four blocks.  The following
problems are perhaps the natural next steps.

\begin{problem}
  What is the complexity of $\Pi'_D$ when $D$ is an oriented cycle
  which is not directed?
\end{problem}

\begin{problem}
 What is the complexity of $\Pi'_D$ when $D$ is an oriented path
 which is not directed?
\end{problem}

Note that the approach used above to find an induced subdivision of
$A^-_4$ relied on the fact that one can check in polynomial time (using
flows) whether a digraph contains internally disjoint $(x,z)$-,
$(y,z)$-paths for prescribed distinct vertices $x,y,z$. If we want to
apply a similar approach for $A^-_5$, then for prescribed vertices
$x,y,z,w$ we need to be able to check the existence of internally
disjoint paths $P,Q,R$ such that $P$ is an $(x,y)$-path, $Q$ is a
$(z,y)$-path and $R$ is a $(z,w)$-path such that these paths are
induced and have no arcs between them. However, the problem of
deciding just the existence of internally disjoint paths $P,Q,R$ with
these prescribed ends is NP-complete by the result of Fortune
et al. \cite{fhw:orientedLink}. Thus we need another approach to obtain
a polynomial-time algorithm (if one exists).

\vspace{12pt}

\begin{figure}[h]
\begin{center}
\scalebox{0.5}{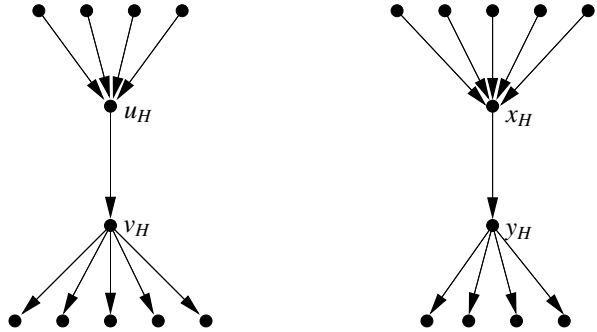}
\end{center}
\caption{The digraph $H$ with specified vertices $u_H,v_H,x_H,y_H$.}\label{immersefig}
\end{figure}

It seems that little is known about detecting a subdivision of some given digraph $D$
as a subgraph  (possibly
non-induced).  
This leads us to the following problem: 

\begin{problem}
  When $D$ is fixed directed graph, what is the complexity of deciding
  whether  a given digraph $G$ contains a $D$-subdivision as a subgraph?
\end{problem}

The following shows that the problem above can be NP-complete.

\begin{theorem}
  \label{th:strongImNPC}
  Let $H$ be the digraph in Figure \ref{immersefig}. It is NP-complete
  to decide whether a given digraph $G$ contains an $H$-subdivision.
  \end{theorem}

\begin{proof}
  By the classical result of Fortune, Hopcroft and Wyllie~\cite{fhw:orientedLink}, the so-called 2-linkage problem
  (given a digraph and four distinct vertices $u,v,x,y$; does $G$
  contain a pair of vertex-disjoint paths $P,Q$ so that $P$ is a
  directed $(u,v)$-path and $Q$ is a directed $(x,y)$-path?) is NP-complete. By
  inspecting the proof (see \cite[Section 10.2]{livre:digraph}) it can
  be seen that the problem is NP-complete even when $G$ has maximum
  in- and out-degree at most 3.  Given an instance $G$ of the
  2-linkage problem with maximum in- and out-degree at most 3 and a
  copy of $H$ we form a new digraph $G_H$ by identifying the vertices
  $\{u,v,x,y\}$ with $\{u_H,v_H,x_H,y_H\}$ in that order. Clearly, if
  $G$ has disjoint directed $(u,v)$, $(x,y)$-paths, then we can use these to
  realize the needed paths from $u_H$ to $v_H$ and from $x_H$ to $y_H$
  (and all other paths are the original arcs of $H$). Conversely,
  suppose there is a subdivision $H^*$ of $H$ in $G_H$. For every $v\in \{u_H, v_H, x_H, y_H\}$, let us denote by
  $v^*$ the vertex corresponding to  $v$ in $H^*$.
  Since  $d^-(u_H)=4,d^+(v_H)=5,d^-(x_H)=5,d^+(y_H)=4$ in $G_H$, we have
  $u^*_H=u_H$ ,$v^*_H=v_H$, $x^*_H=x_H$,  and $y^*_H=y_H$. Thus the two disjoint directed
  $(u^*_H,v^*_H)$- $x^*_H,y^*_H)$-paths in $H^*$ are disjoint
  directed $(u,v)$, $(x,y)$-paths in $G$ implying that $G$ is a 'yes'-instance.
\end{proof}

Finally, we would like to point out that in all detection problems
about induced digraphs, backward arcs of paths play an important role,
especially in NP-completeness proofs.  Also, these backward arcs make
all ``connectivity-flavoured'' arguments fail: when two vertices $x, y$
are given, it is not possible to decide whether $x$ can be linked to
$y$.  So, maybe another notion of induced subdigraph containment would
make sense: chords should be kept forbidden between the different
directed paths that arise from subdividing arcs, but backward arcs
inside the paths should be allowed.

\section*{Acknowledgement}

The authors would like to thank Joseph Yu for stimulating discussions.

\end{document}